%% file: ms.tex
\documentclass[conference,a4paper]{IEEEtran}

\usepackage[utf8]{inputenc} 
\usepackage[T1]{fontenc}
\usepackage{url}
\usepackage{ifthen}
\usepackage{cite}
\usepackage[cmex10]{amsmath} 

\input{src_arxiv/packages.sty}

\begin{document}

\title{Attack on the $\edonk$ \\ Key Encapsulation Mechanism}

\author{\IEEEauthorblockN{Matthieu Lequesne}
  \IEEEauthorblockA{
    Sorbonne Universit\'e, UPMC Univ Paris 06\\
    Inria, Team SECRET,\\
    2 rue Simone Iff, CS 42112,\\
    75589 Paris Cedex 12, France\\
    Email: matthieu.lequesne@inria.fr}
  \and
  \IEEEauthorblockN{Jean-Pierre Tillich}
  \IEEEauthorblockA{
    Inria, Team SECRET,\\
    2 rue Simone Iff, CS 42112,\\
    75589 Paris Cedex 12, France\\
    Email: jean-pierre.tillich@inria.fr}
}

\maketitle

\begin{abstract}
\input{src_arxiv/abstract}
\end{abstract}

\input{src_arxiv/1-intro.tex}
\input{src_arxiv/2-rankmetric.tex}
\input{src_arxiv/3-edonk.tex}
\input{src_arxiv/4-attack.tex}
\input{src_arxiv/5-reconstructingH.tex}

\input{src_arxiv/6-decoding.tex}
\input{src_arxiv/7-secret.tex}
\input{src_arxiv/8-conclusion.tex}

\enlargethispage{-10cm} 

\bibliography{src_arxiv/codecrypto}
\bibliographystyle{alpha}


\end{document}

%% file: src_arxiv/abstract.tex
The key encapsulation mechanism $\edonk$ was proposed in response to the call for post-quantum cryptography standardization issued by the National Institute of Standards and Technologies (NIST). This scheme is inspired by the McEliece scheme but uses another family of codes defined over $\mathbb{F}_{2^{128}}$ instead of $\F$ and is not based on the Hamming metric. It allows significantly shorter public keys than the McEliece scheme.

In this paper, we give a polynomial time algorithm that recovers the encapsulated secret. This attack makes the scheme insecure for the intended use. We obtain this result by observing that recovering the error in the McEliece scheme corresponding to $\edonk$ can be viewed as a decoding problem for the rank-metric. We show that the code used in $\edonk$ is in fact a super-code of a Low Rank Parity Check (LRPC) code of very small rank (1 or 2). 
A suitable parity-check matrix for the super-code of 
such low rank can  be easily derived from for the public key. We then use this parity-check matrix in a decoding algorithm that was devised for LRPC codes to recover the error. Finally we explain how we decapsulate the secret once we have found the error.

%% file: src_arxiv/1-intro.tex
\section{Introduction}

The syndrome decoding problem is a fundamental problem in complexity theory, since the original paper of Berlekamp, McEliece and van Tilborg \cite{BMT78} proving its NP-completeness for the Hamming distance. The same year, McEliece proposed a public-key cryptosystem based on this problem \cite{M78} and instantiated it with binary Goppa codes. This scheme was for a long time considered inferior to RSA due to its large key size. However, this situation has changed drastically
when it became apparent in \cite{S94a} that RSA and actually all the other public-key cryptosystems used in practice could be attacked in polynomial time
by a quantum computer. There are now small prototypes of such computers that lead to think that they will become a reality in the future and 
in 2016, the National Institute of Standards and Technology (NIST) announced a call for standardization of cryptosystems that would be safe against an adversary equiped with a quantum computer. Four families of cryptosystems are often mentioned as potential candidates: cryptosystems based on error correcting codes, lattices, hash functions and multivariate quadratic equations \cite{BBD08}. All of these are based on mathematical problems that are expected to remain hard even in the presence of a quantum computer.

The key encapsulation mechanism $\edonk$ \cite{GG17} was proposed by Gligoroski and Gj{\o}steen in response to the call issued by the NIST. This scheme is inspired by the McEliece scheme but uses another family of codes defined over $\mathbb{F}_{2^{128}}$ instead of $\F$. This choice leads to very short keys for a code-based scheme. The metric used for the decoding is not properly defined and the security relies on an ad-hoc problem named \textit{finite field vector subset ratio  problem} supposedly hard on average.

In this paper, we show that the metric used for $\edonk$ is in fact equivalent to the well-known rank metric. This metric was first introduced 
in 1951 as ``arithmetic distance'' between matrices over a field $\Fq$ \cite{H51}. The notion of rank distance and rank codes over matrices was defined in 1978 by Delsarte \cite{D78}. 
He introduced a code family, named maximum rank distance (MRD) codes, that attains the analogue of the MDS (maximum distance separable) bound for the rank metric. 
Gabidulin suggests in \cite{G85} to consider a subfamily of such codes that are linear over an extension field $\Fqm$. It provides 
a vectorial representation of these codes and allows to represent them in a much more compact way. This is the main reason why 
the rank metric based McEliece schemes achieve significantly smaller key sizes. Moreover this vectorial representation allows to view the 
known families of MRD codes as rank metric analogues of Reed-Solomon codes and to obtain an efficient decoding algorithm for them \cite{G85}.
There are also rank metric analogues for other families of codes. For instance, the Low Rank Parity-Check (LRPC) codes introduced in \cite{GMRZ13} can be considered as 
analogues of Low Density Parity-Check (LDPC) codes. Just like their binary cousins, they enjoy an efficient decoding algorithm that is based on a low rank parity-check matrix of such a code.
 
 Here, we prove that the code used in $\edonk$ is a actually a super-code of an LRPC code of rank $2$. What is more, this LRPC code is itself a subspace of codimension 1 of another LRPC code of rank 1. It turns out that parity-check matrices of rank 2 for the first super-code and rank $1$ for the second one can easily be derived from the public key. In both cases, this allows us to decode the ciphertext without the secret key. This gives a way to recover the encapsulated secret and breaks completely the $\edonk$ system.

The paper is organized as follows.
First, we recall some basic definitions and properties of rank metric and LRPC codes in Section \ref{rankmetric}. In Section \ref{edonk} we present the scheme of $\edonk$. Then we explain the general idea of our attack in section \ref{attack}. In Section \ref{reconstructing}, we detail how we reconstruct a parity-check matrix of the code and in Section \ref{decoding} how we decode the ciphertext. In Section \ref{secret}, we explain how we derive the encapsulated secret from the error. Finally in Section \ref{conclusion} we discuss the cost of this attack and its consequences.

%% file: src_arxiv/2-rankmetric.tex
\section{Rank metric codes}
\label{rankmetric}

\subsection{Notation}

In the following document, $q$ denotes a power of a prime number. 
In the case of $\edonk$, we will have $q=2$.
$\Fq$ denotes the finite field with $q$ elements and, for any positive integer $m$, $\Fqm$ denotes 
the finite field with $q^m$ elements. We will sometimes view $\Fqm$ as an $m$-dimensional vector space over $\Fq$.

We use bold lowercase and capital letters to denote vectors and matrices respectively.

We denote $\langle x_1, \ldots, x_k \rangle_{\mathbb{K}}$ the $\mathbb{K}$-vector space generated by the elements $\{x_1, \ldots x_k\}$.

\subsection{Definitions}

\begin{definition}[Rank metric over $\Fqm^n$]
Let $\x = (x_1, \ldots, x_n)\in\Fqm^n$ and $(\beta_1, \ldots, \beta_m)$ be 
a basis of $\Fqm$ viewed as an $m$-dimensional vector space over $\Fq$. Each coordinate $x_j \in \Fqm$ is associated to a vector of $\Fq^m$ in this basis: 
$x_j = \sum_{i=1}^mm_{i,j}\beta_i$.
 The $m \times n$ matrix associated to $\x$ is given by $\M(\x) \defeq (m_{i,j})_{1\leq i \leq m, 1 \leq j \leq n}$.

The rank weight $\w(\x)$ of $\x$ is defined as :
\[\w(\x) \defeq \mathrm{Rank\,} \M(\x).\]

The associated distance $d(\x,\y)$ between elements $\x$ and $\y$ of $\Fqm^n$ is defined by $d(\x,\y) \defeq \w(\x-\y)$.
\end{definition}

\begin{definition}[Support of a word]
Let $\x = (x_1, \ldots, x_n) \in \Fqm^n$. The support of $\x$, denoted $\mathrm{Supp}(\x)$, is the $\Fq$-subspace of $\Fqm$ generated by the coordinates of $\x$:
\[ \mathrm{Supp}(\x) \defeq \vsg{x_1, \ldots, x_n}.\]
We have $\mathrm{dim}(\mathrm{Supp}(\x))=\w(\x).$
\end{definition}

\begin{definition}[$\Fqm$-linear code]
An $\Fqm$-linear code $\C$ of dimension $k$ and length $n$ is a subspace of dimension $k$ of $\Fqm^n$.
$\C$ can be represented in two equivalent ways: by a generator matrix $\G \in \Fqm^{k \times n}$ such that $\C = \{\x \G \st \x \in \Fqm^k \}$ and by a parity-check matrix $\H \in \Fqm^{(n-k) \times n}$ such that $\C = \{\x \in \Fqm^n \st \H\transpose{\x} = \mathbf{0}_{n-k}\}$.
\end{definition}

The decoding problem in the rank metric can be described as follows.
\begin{problem}[Decoding problem for the rank metric]
Let $\C$ be an $\Fqm$-linear code of dimension $k$ and length $n$. Given $\y = \c + \e$ where $\c \in \C$ and $\e \in \Fqm^n$ is of rank weight $\leq r$ find $\c$ and $\e$.
\end{problem}

\subsection{LRPC codes}

\begin{definition}[LRPC code]
\label{lrpc}
A Low Rank Parity Check (LRPC) code of rank $d$, length $n$ and dimension $k$ over $\Fqm$ is a code that admits a parity-check matrix 
$\H = (h_{i,j}) \in \Fqm^{(n-k) \times n}$ 
such that the vector space of $\Fqm$ generated by its coefficients $h_{i,j}$ has dimension at most $d$.
\end{definition}

LRPC codes can be viewed as analogues of LDPC codes for the rank metric. In particular, they enjoy an efficient decoding algorithm based on their low rank parity-check matrix. Roughly speaking, Algorithm 1 of \cite{GMRZ13} decodes up to $d$ errors when $rd \leq n-k$ in polynomial time (see \cite[Theorem 1]{GMRZ13}).
It uses in a crucial way the notion of the linear span of a product of subspaces of $\Fqm$
\begin{definition}Let $U$ and $V$ be two $\Fq$ subspaces of $\Fqm$. We denote by $U \cdot V$ the linear span of the 
product of $U$ and $V$:
$$
U \cdot V \eqdef \vsg{uv:u \in U,\;v \in V}.
$$
\end{definition}


%% file: src_arxiv/3-edonk.tex
\section{The $\edonk$ KEM}
\label{edonk}

$\edonk$ \cite{GG17} is a key encapsulation mechanism proposed by Gligoroski and Gj{\o}steen for the NIST post-quantum cryptography call. Here we describe the key generation, encapsulation and decapsulation, omitting some details that are not relevant for the attack. We refer to \cite{GG17} for the full description.

\subsection{Parameters and notations}

The parameters for $\edonk$ are given in Table \ref{tab:parameters}. In this paper we often refer to the parameters of $\edonkref$, the reference version proposed for 128 security-bits.

\begin{table}[htbp]
  \renewcommand{\arraystretch}{1.3}
  \caption{Parameters proposed for $\edonk$}
  \label{tab:parameters}
  \centering
  \begin{tabular}{|l|c|c|c|c|c|c|}
    \hline
    Name & $m$ & $N$ & $K$ & $R$ & $\nu$ & $L$ \\
    \hline\hline
    $\boldsymbol{\mathsf{edonk128ref}}$ & \textbf{128} & \textbf{144} & \textbf{16} & \textbf{40} & \textbf{8} & \textbf{6} \\    
    \hline
    $\mathsf{edonk128K16N80nu8L6}$ & 128 & 80 & 16 & 40 & 8 & 6 \\
    \hline
    $\mathsf{edonk128K08N72nu8L8}$ & 128 & 72 & 8 & 40 & 8 & 8 \\
    \hline
    $\mathsf{edonk128K32N96nu4L4}$ & 128 & 96 & 32 & 40 & 4 & 4 \\
    \hline
    $\mathsf{edonk128K16N80nu4L6}$ & 128 & 80 & 16 & 40 & 4 & 6 \\
    \hline\hline
    $\mathsf{edonk192ref}$ & 192 & 112 & 16 & 40 & 8 & 8 \\
    \hline
    $\mathsf{edonk192K48N144nu4L4}$ & 192 & 144 & 48 & 40 & 4 & 4 \\
    \hline
    $\mathsf{edonk192K32N128nu4L6}$ & 192 & 128 & 32 & 40 & 4 & 6 \\
    \hline
    $\mathsf{edonk192K16N112nu4L8}$ & 192 & 112 & 16 & 40 & 4 & 8 \\
    \hline
  \end{tabular}
\end{table}

The scheme makes use of a hash function $\hash{\cdot}$ corresponding to standard SHA2 functions (SHA-256 or SHA-384 depending on the parameters). We will denote $\mathcal{H}^i(\cdot) \defeq \underbrace{\mathcal{H}( \ldots \mathcal{H}  (\cdot) )}_{i \text{ times}}$

Given a binary matrix $\P = (p_{i,j})$ and two non-zero elements $a \neq b$ of $\Fm$, $\P_{a,b} = (\tilde{p}_{i,j})$ denotes the matrix of the same size with coefficients in $\Fm$ where $\tilde{p}_{i,j} = a$ if $p_{i,j}=0$ and $\tilde{p}_{i,j} = b$ if $p_{i,j}=1$.

In particular, if $\P$ is orthogonal:
\begin{equation}
\label{eq:orthogonal}
\transpose{\P_{c,d}} = (\P_{a,b})^{-1} 
\end{equation}
 where $c \defeq \frac{a}{a^2 + b^2}$ and $d \defeq \frac{b}{a^2 + b^2}$.

For two vectors (or matrices) $\mathbf{x}$ and $\mathbf{y}$, we will denote $\mathbf{x} \vert\vert \mathbf{y}$ their concatenation.

\subsection{Key generation}

Given the security level and the appropriate parameters.

\begin{itemize}
\item $a, b \getrand \Fm$ non-zero elememts such that $a \neq b$.
\item $\P \getrand \F^{N \times N}$ an orthogonal matrix.
\item $\H \getrand \F^{R \times N}$ such that $\H = \transpose{[\H_T \vert\vert \H_B]}$ where $\H_B$ is an $R \times R$ orthogonal matrix and $\H_T$ has columns of even Hamming weight. 
\item $c \defeq \frac{a}{a^2 + b^2}$, $d \defeq \frac{b}{a^2 + b^2}$.
\item $\tg \getrand \Fm^{\nu}$.
\item $\V_g \defeq \Sup{\tg}$.
\item $\G \getrand \V_g^{K \times N}$ such that 
\begin{equation}
\label{eq:GH}
\G  \transpose{\H} = \mathbf{0}_{K \times R}.
\end{equation}
\item $\Gpub \defeq \G  \transpose{\P_{c,d}}$.
\item Return ($\pk \defeq \Gpub, \sk \defeq (a,b,\P,\H)$).
\end{itemize}

\subsection{Encapsulation}

Given the $\pk$ and the public parameters. 
\begin{itemize}
\item $\m \getrand \Fm^{K}$.
\item $\te \in \Fm^L$ generated as follows:
  \begin{itemize}
  \item $(\tilde{e}_{0},\tilde{e}_{1}) \getrand \Fm$;
  \item for $1 \leq i \leq \frac{L}{2}-1$, $(\tilde{e}_{2i},\tilde{e}_{2i+1}) = \hash{\tilde{e}_{2i-2} \vert\vert \tilde{e}_{2i-1}}$.
  \end{itemize}
\item $\V_e \defeq \Sup{\te}$.
\item $\e \getrand \V_e^{N}$.
\item $\c \defeq \m  \Gpub + \e$.
\item $(s_0,s_1) \defeq \hash{\tilde{e}_{L-2}\vert\vert \tilde{e}_{L-1}}$.
\item $\secret \defeq \hash{s_0 \vert\vert s_1 \vert\vert \hash{\c}}$.
\item $h \defeq \hash{s_1 \vert\vert s_o \vert\vert \hash{\c}}$.
\item $\ctxt \defeq (\c, h)$.
\item Return $(\ctxt, \secret)$.
\end{itemize}

\subsection{Decapsulation}

Given $\ctxt$, $\sk$ and the public parameters. 
\begin{itemize}
\item Recover $\e$ by decoding the $\c$ using the private matrix $\H' \defeq \H\transpose{\P_{a,b}}$.
\item Deduce $\V_e$ the vector space spaned by the coefficients of the vector $\e$.
\item For all $(\lambda, \nu) \in \V_e \times \V_e$, for $1 \leq i \leq \frac{L}{2}-1$: 
  \begin{itemize}
  \item $(s'_0, s'_1) \defeq \mathcal{H}^i\left(\lambda \vert\vert \mu \vert\vert \hash{\c}\right)$;
  \item if $\hash{s'_1 \vert\vert s'_0 \vert\vert c} = h$:
    \begin{itemize}
    \item[] Return $\secret \defeq \hash{s'_0 \vert\vert s'_1 \vert\vert \hash{\c}}$.
    \end{itemize}
  \end{itemize}
\end{itemize}

%% file: src_arxiv/4-attack.tex
\section{Outline of the Attack on $\edonk$}
\label{attack}

Our attack is based on three observations
\begin{itemize}
\item The ciphertext is a vector $\c$ such that 
\begin{equation}
\label{eq:decoding}
\c = \m  \Gpub + \e.
\end{equation} This error $\e$ is of low rank, since its rank is at most $L$.
\item This code $\Cpub$ generated by $\Gpub$ is a subcode of an LRPC code, namely the code $\C'$ with parity-check matrix $\tH \defeq \H \transpose{\P_{a,b}}$. This code is indeed  an LRPC code of rank $2$ since all the entries of $\tH$ belong to 
$\bvsg{a,b}$. We have 
\begin{equation}
\label{eq:CpubCp}
\Cpub \subset \C'
\end{equation} since 
\begin{eqnarray*}
\Gpub \transpose{\tH} &= & \G \transpose{\P_{c,d}}\transpose{\left(\H \transpose{\P_{a,b}}\right)}\\
& = & \G \transpose{\P_{c,d}}\P_{a,b}\transpose{\H}\\
& = & \G \transpose{\H} \;\;\text{(from \eqref{eq:orthogonal})}\\
& = & \mathbf{0}_{K \times R} \;\;\text{(from \eqref{eq:GH})}.
\end{eqnarray*}
This equation also appears as Corollary 1  of  \cite[p.19]{GG17}. We have given its proof here for the convenience of the reader.
Let $K'=N-R$ be the dimension of $\C'$.
\item If we recover a parity-check matrix of rank $2$ for $\C'$ we will be able to recover $\m \Gpub$ and $\e$ from $\c$. 
Indeed, $\m \Gpub \in \C'$ and we can decode $\C'$ using a variation of Algorithm 1 of \cite{GMRZ13} and the knowledge of the parity-check matrix, provided $\w(\e) \leq L < (N-K')/2=R/2$ is verified, which is the case for the parameters of $\edonk$.
\end{itemize}

Hence we will proceed in three steps:
\begin{enumerate}
\item constructing and solving a linear system of equations to find a parity-check matrix for the code $\C'$ (detailed in Section \ref{reconstructing});
\item decoding the ciphertext using a slight variation of Algorithm 1 of \cite{GMRZ13} (see Section \ref{decoding}); 
\item recovering the secret from the error vector (explained in Section \ref{secret}).
\end{enumerate}

%% file: src_arxiv/5-reconstructingH.tex
\section{Reconstructing the parity-check matrix}
\label{reconstructing}

\subsection{Compressed public key}

In order to reduce the public key size, the designers of $\edonk$ chose to represent the public key in a compressed form. 
They took advantage of the fact that all the coefficients of $\Gpub$ live in the vector space
$\V_{g,c,d} \defeq \bvsg{c\tilde{g}_1, \ldots, c\tilde{g}_\nu, d\tilde{g}_1, \ldots, d\tilde{g}_\nu}$ of dimension $2\nu$.
Hence, the compressed public key consists in two parts: first the basis $\tg_{c,d} \defeq (c\tilde{g}_1, \ldots, c\tilde{g}_\nu, d\tilde{g}_1, \ldots, d\tilde{g}_\nu) \in \Fm^{2\nu}$ of the vector-space $\V_{g,c,d}$, then the entries of the matrix $\Gpub$ such that each entry is represented by its coefficients in the basis $\tg_{c,d}$. For example, if an entry $x$ of $\Gpub$ is equal to $c \sum_{i=1}^{\nu} \gamma_i \tilde{g}_i + d \sum_{i=1}^{\nu} \delta_i \tilde{g}_i$ with $\gamma_i, \delta_i \in \F$, $x$ will be represented by $(\gamma_1, \ldots, \gamma_{\nu}, \delta_1, \ldots, \delta_{\nu}) \in \F^{2\nu}$. There is another subtlety in the compression that we will not mention here.

\subsection{Finding a basis}

The attacker does not have access to the value of $a$ and $b$ but can deduce the value of $ab^{-1} = cd^{-1}=(c\tilde{g}_1)(d\tilde{g}_1)^{-1}$ from $\tg_{c,d}$ as mentioned in paragraph 7.2.2 of the documentation of $\edonk$ \cite{GG17}.

Let us bring in
\[\alpha \defeq ab^{-1}.\]

We notice that $\H" \defeq b^{-1}\tH$ is also a parity-check matrix of the LRPC code $\C'$.
This matrix has all its coefficients in $\bvsg{1,\alpha}$.
We use this information to reconstruct such a parity-check matrix of the code $\C'$ by solving a linear system, similarly to what is done in \cite[Section IV B]{GRS16}. This system is derived from the following facts:
\begin{itemize}
\item[(i)]
 $\Gpub\;\transpose{\H''}=\mathbf{0}_{K \times R}$;
 \item[(ii)]
 the entries of $\H''$ belong to $\bvsg{1,\alpha}$.
 \end{itemize}
In other words, the possible rows $\x=(x_1,\dots,x_N)$ of $\H''$ are solutions  of the following system
\begin{equation}
\label{eq:linear}
\left\{ \begin{array}{lcl}
\Gpub \transpose{\x} & = & \mathbf{0}_{K}\\
x_i & \in &\bvsg{1,\alpha}\;\text{for all $i \in \{1,\dots,N\}$}.
\end{array}
\right.
\end{equation}
This system is obviously linear over $\F$ and the solution set is an $\F$-linear subspace.
A basis of this subspace can then be used as rows for $\H''$.
We now show that solving this system can be done by solving a linear system over $\F$.

\subsection{Recovering $\H''$ by solving a linear system over $\F$ and an affine system in a more general case}
\label{ss:solving}

Actually in this section we will consider a more general version of \eqref{eq:linear}. Given a system
\begin{equation}
\label{eq:affine}
\A \transpose{\x} = \transpose{\b}
\end{equation}
where $A=(a_{ij})_{{1 \leq i \leq r},{1 \leq j \leq N}}$ is a given matrix in $\Fm^{r\times N}$ and
$\b$ is a given vector in $\Fm^r$, 
and given $V$ a subspace of dimension $t$ of $\Fm$ (viewed as vector space over $\F$ of dimension $m$),
how to find the affine set of the solutions  $\x=(x_i)_{1 \leq i \leq N} \in V^N$ of the system?

We can rewrite the system \eqref{eq:affine} as
\begin{equation}
\left\{ \begin{array}{rcl}
a_{11}x_1 + \cdots + a_{1N}x_N & = & b_1\\
\cdots & = & \cdots \\
a_{r1}x_1 + \cdots + a_{rN}x_N & = & b_r.
\end{array}
\right.
\end{equation}

We introduce a basis $\{v_1,\dots,v_t\}$ of $V$ and express each unknown $x_j$ in this basis in terms of $t$ other unknowns $x_{j1},\dots,x_{jt} \in \F$:
\[ x_j = \sum_{i=1}^t x_{ji} v_i.\]

In other words, the system \eqref{eq:affine} is equivalent to
\begin{equation}
\label{eq:affine2}
\left\{ \begin{array}{rcl}
\sum_{j=1}^N\sum_{i=1}^t a_{1j} v_i x_{ji}  & = & b_1\\
\dots & = & \dots \\
\sum_{j=1}^N\sum_{i=1}^t a_{rj} v_i x_{ji} & = & b_r.
\end{array}
\right.
\end{equation}

Let $\{\beta_1, \dots, \beta_{m}\}$ be an $\F$-basis of $\Fm$, 
we introduce for $1 \leq \ell \leq m$ the projection $\pi_\ell$ from $\Fm$ to $\F$ defined by:
\begin{equation}
\pi_\ell: 
\begin{array}{ccc} 
\Fm & \longrightarrow & \F \\
a = \sum_{j=1}^{m} a_j \beta_j & \longmapsto & a_\ell.
\end{array}
\end{equation}

The $r$ equations of system \eqref{eq:affine2} defined over $\Fm$ lead to $rm$ affine equations over $\F$
by applying $\pi_\ell$ for $\ell \in \{1,\dots,m\}$: 

\begin{equation}
\left\{ \begin{array}{rcl}
\sum_{j=1}^N\sum_{i=1}^t \pi_\ell(a_{1j} v_i) x_{ji}  & = & \pi_{\ell}(b_1)\\
\dots & = & \dots \\
\sum_{j=1}^N\sum_{i=1}^t \pi_\ell(a_{rj} v_i) x_{ji} & = & \pi_\ell(b_r).
\end{array}
\right.
\end{equation}

We can solve this affine system in $\F$ to recover the solution of \eqref{eq:affine}. The system has $rm$ binary equations and $tN$ unknowns, hence a complexity of $\O(rmt^2N^2)$. 
If we apply this technique to \eqref{eq:linear}, where $t=2$ and $r=K$ we obtain a basis of the vector space in 
time $O(KmN^2)$.

%% file: src_arxiv/6-decoding.tex
\section{Decoding step}
\label{decoding}

The previous step recovers an $R \times N$ matrix $\H^{(3)}$ whose entries all belong
to $\bvsg{1,\alpha}$. The matrices $\H^{(3)}$ and $\H''$ share the property that their rows
form a basis of solutions of \eqref{eq:linear}. Therefore, there exists an $R \times R$ binary 
invertible matrix $\Q$ such that
\begin{equation}
\label{eq:H3Hsec}
\H^{(3)} = \Q\H''.
\end{equation}

We use $\H^{(3)}$ to decode and recover $\e$ from the ciphertext $\c$.
The vectors are linked by the equation
\begin{equation}
\c = \m  \Gpub + \e.
\end{equation}
We use here a slight variation of Algorithm 1 of \cite{GMRZ13} to decode. 
Algorithm 1 would consist in performing the following steps:
\begin{enumerate}
\item Compute $\transpose{\s} \defeq \H^{(3)}\transpose{\c}$ and then $V \defeq \Sup{\s}$.
Here we  typically have $V = \Sup{\e}\cdot \vsg{1,\alpha}$ when $\H^{(3)}$ is a random matrix.
\item Compute $V' \defeq V \cap \alpha^{-1}V$. This step typically recovers $\Sup{\e}$ when
$V = \Sup{\e}\cdot \vsg{1,\alpha}$. 
\item Once we have $\Sup{\e}$ we recover $\e=(e_1,\dots,e_N)$ by solving the linear equation 
$\H^{(3)}\transpose{\e}=\transpose{\s}$ with the additional constraints $e_i \in \Sup{\e}$ for 
$i \in \{1,\ldots,N\}$. This is done by using the technique given in Subsection \ref{ss:solving}.
\end{enumerate}
In our case, due to the special structure of $\H$ which contains only $a$'s and $b$'s $V$ is not 
equal to $\Sup{\e}\cdot \vsg{1,\alpha}$. This is due to the following result.

\begin{proposition}
We have for every $\e \in \Fm^N$:
\[\Sup{\H^{(3)}\transpose{\e}} \subset (1+\alpha)\Sup{\e}+\bvsg{\sum_{i=1}^N e_i}.\]
\end{proposition}

\begin{proof}
From \eqref{eq:H3Hsec}, we deduce that
\[\Sup{\H^{(3)}\transpose{\e}} = \Sup{\Q\H''\transpose{\e}} = 
\Sup{\H''\transpose{\e}}.\]

Let $\transpose{\s} \eqdef \H''\transpose{\e}$. Denote the $i$-entry of $\s$ by $s_i$ and
the entry of $\H''$ in row $i$ and column $j$ by $h''_{ij}$. We have:
\begin{eqnarray*}
s_i & = & \sum_{j=1}^N h''_{ij}e_j\\
    & = & \sum_{j \text{ s.t. }h''_{ij}=1} e_j + \sum_{j \text{ s.t. } h''_{ij}=\alpha}\alpha e_j\\
    & = & \sum_{j=1}^N e_j + (1+\alpha) \sum_{j \text{ s.t. } h''_{ij}=\alpha} e_j.
\end{eqnarray*}
This implies the proposition.
\end{proof}

This proposition directly gives a subspace of dimension ${L+1}$ that contains
$\Sup{\e}$ since we deduce from it that
\begin{equation}
\label{eq:fundamental}
\Sup{\e} \subset (1+\alpha)^{-1} \Sup{\H'' \e}.
\end{equation}

A slight modication of Algorithm 1 of \cite{GMRZ13} yields therefore $\e$: 
\begin{enumerate}
\item compute the syndrome $\transpose{\s} \defeq \H^{(3)}\transpose{\c}$ and then $V \defeq (1+\alpha)^{-1} \Sup{\s}$;
\item The space $V$ contains $\Sup{\e}$, so we can recover $\e=(e_1,\dots,e_N)$ by solving the linear equation 
$\H^{(3)}\transpose{\e}=\transpose{\s}$ with the additional constraints $e_i \in V$ for 
$i \in \{1,\ldots,N \}$. This is done by using the technique given in Subsection \ref{ss:solving}.
\end{enumerate}
 
Note that we can also skip step 2 and directly look for $s_0$ and $s_1$ in the space $V$ of dimension $L+1$ instead of decoding exactly the value of $e$. In fact, this is what is specified in the decapsulation of $\edonk$.

%% file: src_arxiv/7-secret.tex
\section{Recovering the shared secret}
\label{secret}

Once we have recovered the error vector $\e \in \Fqm^N$, we need to recover $s_0$ and $s_1$ to obtain the value of $\secret$. We know that the elements of $\e$ were picked randomly in $\V_e = \Sup{\te}$.We proceed just like in the decapsulation algorithm.

We generate $\Sup{\e}$ which is equal to $\V_e$ with high probability. More exactly, the probability that $\Sup{\e}$ is of dimension $< L$ is $\left(\frac{L-1}{L}\right)^{N}$. For the parameters of $\edonkref$ this probability is $2^{-37}$. In such a case, the attack might fail, but the decapsulation would fail too.

Then, among the $2^L$ elements of $\V_e$, we need to identify a couple of consecutive elements of $\te$ to deduce the secret. For all pairs of candidates $(\lambda, \mu) \in \V_e \times \V_e$, for $1 \leq i \leq \frac{L}{2}-1$ we compute $(s'_0, s'_1) \defeq \mathcal{H}^i\left(\lambda \vert\vert \mu \vert\vert \hash{\c}\right)$. If $\hash{s'_1 \vert\vert s'_0 \vert\vert c} = h$ then we have $(s'_0,s'_1) = (s_0, s_1)$. Finally we recover the secret $\secret = \hash{s_0 \vert\vert s_1 \vert\vert \c}$. In total this operation requires $\O(L2^{2L})$ operations, just like the decapsulation. This is the reason why the value of $L$ needs to remain small, otherwise the decapsulation is not possible.

%% file: src_arxiv/8-conclusion.tex
\section{Concluding remarks}
\label{conclusion}

\subsection{Cost of the attack}

Let us analyze the cost of the three steps of the attack mentioned in Section \ref{attack}.

Step 1 and 2 are polynomial in terms of the parameters of the code. Step 1 only uses linear algebra operations and has a complexity at most $\O(KmN^2)$. The complexity of step 2 is given by Theorem 1 of \cite{GMRZ13} (using $n=N, k=N-R, r=L$ and $d=2$), hence is equal to $L^2(16m+N^2)$. The complexity of step 3 is $\O(L2^{2L})$. This is not polynomial in $L$ but $L$ is a very small parameter ($4 \leq L \leq 8$ in the proposal). Moreover this third step is the same as the decapsulation algorithm, so $L$ needs to stay small, otherwise the decapsulation would become too costly or even impossible. So $L$ can be considered as a constant $\leq 10$ to allow a reasonable decapsulation. Hence the most costly operation appears to be step 1.

\subsection{Without compression of the public key}

Our attack takes advantage of the compressed form of the public key that allows a direct access to the value $\alpha = ab^{-1}$. One could think that this is the origin of the attack, and decide to express the public key in its uncompressed form to fix the attack. As a consequence, the public key would be of size $K \times N \times m$ bits instead of $K \times N \times \nu$ bits in the compressed form. In practice the public key for $\edonkref$ would be  16 times longer (around 288 kbits). This inflation of the key size could be avoided by sending out a random basis of the space $\V_{g,c,d}$.

However, this is not enough. There is an even more direct way to proceed, without the value of $\alpha$. Instead of looking for a matrix $\H^{(3)}$ with entries liyng in $\bvsg{1,\alpha}$, we can use the following result. 
\begin{proposition}
There exists a full rank $(R-1) \times N$ binary matrix $\H^{(4)}$ that satisfies 
\[\Gpub \transpose{\H^{(4)}}=\mathbf{0}_{K \times (R-1)}.\]
\end{proposition}

\begin{proof}
Let $\T$ be a binary full-rank matrix $(R-1) \times R$ matrix that has rows of even Hamming weight. 
For instance we can choose
\[
\T = \begin{pmatrix}
1 & 1 & 0 & \cdots & 0 \\
0 & 1 & 1 & 0 & \vdots \\
\vdots & \ddots & \ddots & \ddots & \vdots\\
0 & \cdots & 0 & 1 & 1
\end{pmatrix}.
\]

We observe now that $\T \H$ has all its entries in $\{0,a+b\}$. This follows directly from the fact 
that if we sum an even number of elements in $\{a,b\}$ we either get $0$ 
(if the number of $a$'s is even, and therefore also the number of $b$'s) or
$a+b$ (if the number of $a$'s is odd).
From this, it follows immediately that
\[\H^{(4)} \defeq \frac{1}{a+b} \T \H\]
satisties the property. First, it is clear that this is a binary matrix and we also have
\begin{eqnarray*}
\Gpub \transpose{\H^{(4)}} & = & \frac{1}{a+b} \Gpub \transpose{\H} \transpose{\T}\\
& = & \mathbf{0}_{K \times (R-1)}.
\end{eqnarray*}
\end{proof}

Obtaining such a matrix $\H^{(4)}$ is straightforward. We just have to use the algorithm given in 
Section \ref{reconstructing} to recover a basis of dimension $R-1$ of binary vectors $\x$ satisfying
\[\Gpub \transpose{\x}=\mathbf{0}_{K}.\]
We then use this matrix $\H^{(4)}$ to compute the syndrome ${\s = \H^{(4)}\transpose{\c}}$. Since 
$\H^{(4)}\transpose{\c}= \H^{(4)}\transpose{\e}$ we directly obtain with very high probability that
\[\Sup{\e}=\Sup{\H^{(4)}\transpose{\c}}.\]

This reveals the support of the error and from there we can go directly to the last step of the attack to reconstruct
the shared secret.

\subsection{Security of the scheme}

Considering the attack that we described, there is a way to recover the secret of the $\edonkref$ scheme from a public key without the private key in polynomial time. In practice, the attack implemented with Sage on a personal computer recovers the secret in less than a minute, so the scheme is far from achieving the 128-bits security claimed in \cite{GG17}. Hence this scheme is insecure for the intended use. Moreover, the cost of this attack is polynomial in terms of the parameters, so there is no proper way to increase the parameters to achieve the intended security level while keeping a reasonably small key size.